\documentclass[12pt]{article}

\usepackage{fullpage}
\usepackage[T1]{fontenc}
\usepackage{ae,aecompl}
\usepackage[dvips]{graphicx}
\usepackage{amssymb}
\usepackage{amsmath}
\usepackage{subfigure}
\usepackage[round,authoryear]{natbib}
\usepackage{color}
\usepackage{array}
\usepackage{rotating}
\usepackage{colortbl}
\usepackage{tikz}
\usetikzlibrary{patterns}
\newcommand{\Qomit}[1]{}
\newcommand{\Xomit}[1]{}

\usepackage[normalem]{ulem}

\definecolor{webgreen}{rgb}{0,0.4,0}
\definecolor{webbrown}{rgb}{0.6,0,0}
\definecolor{purple}{rgb}{0.5,0,0.25}
\definecolor{darkblue}{rgb}{0,0,0.7}
\definecolor{darkred}{rgb}{0.7,0,0}
\definecolor{darkgreen}{rgb}{0,0.7,0}
\usepackage[pdfborder=false]{hyperref}
\hypersetup{colorlinks,citecolor=darkred,filecolor=black,linkcolor=darkblue,urlcolor=webgreen,pdfpagemode=none,
pdfstartview=FitH}
\newcommand{\ignore}[1]{}

\newtheorem{prop}{{\sc Proposition}}
\newtheorem{cor}{{\sc Corollary}}
\newtheorem{theorem}{{\sc Theorem}}
\newtheorem{defn}{{\sc Definition}}

\newtheorem{example}{{\sc Example}}

\newenvironment{proof}{\noindent {\bf \sl Proof\/}:\enspace}
{\hfill $\blacksquare{}$ \vspace{12pt}}
\usepackage{sectsty}
\sectionfont{\centering\normalfont\scshape}
\subsectionfont{\centering\normalfont}

\begin{document}

\begin{titlepage}
\title{\textbf{Symmetric reduced form voting}\thanks{We are grateful to Federico Echenique, Arunava Sen, Zaifu Yang,
and two anonymous referees for their comments.}}

\author{Xu Lang\thanks{Southwest University of Finance and Economics ({\tt langxu@swufe.edu.cn}).}$\;\;$ and Debasis Mishra\thanks{Indian Statistical Institute, Delhi ({\tt dmishra@isid.ac.in}).}}

\maketitle

\begin{abstract}
We study a model of voting with two alternatives in a symmetric environment. We characterize the interim
allocation probabilities that can be implemented by a symmetric voting rule. We show that every such
interim allocation probabilities can be implemented as a convex combination of two families of
deterministic voting rules: {\sl qualified majority} and {\sl qualified anti-majority}. We also provide analogous results
by requiring implementation by a symmetric monotone (strategy-proof) voting rule and by a symmetric unanimous voting rule.
We apply our results to show that an {\sl ex-ante} Rawlsian rule is a convex combination of a pair
of qualified majority rules.

\medskip


\end{abstract}
\thispagestyle{empty}
\end{titlepage}

\section{Introduction}

In many mechanism design problems, the incentive constraints and the objective function of the designer
can be written in the interim allocation space. While a mechanism describes the {\sl ex-post} allocation
of the agents, the solution to an incentive constrained optimization may describe only interim allocations.
This raises a natural question:
{\sl which interim allocations can be generated by a (ex-post) mechanism}?
If there is a characterization of interim allocations that can be generated by a mechanism, then it can be
put as a constraint in any incentive constrained optimization. This approach to mechanism design is known as
the reduced form approach. It was pioneered in the single object auction literature by \citet{M84,MR84}, leading
to the seminal characterization in Border's theorem \citep{B91}.

We analyze reduced form voting mechanisms in a simple model of voting with two alternatives: $a$ and $b$.
In our model, each agent has two possible types: (i) $a$-type agent prefers $a$ followed by $b$ and (ii) $b$-type
agent prefers $b$ followed by $a$. We consider a symmetric voting environment: the probability of two type profiles with the same number of $a$-types is identical. Hence, we focus on symmetric voting rules, which  choose a probability distribution
over $a$ and $b$ for every number of $a$-types. The {\sl interim} allocation probability of choosing $a$ (and $b$)
for $a$-type and $b$-type agents can be computed from the symmetric voting rule. The reduced form
voting question is the following: {\sl given the interim allocation probabilities of choosing $a$ and $b$ for $a$-type and $b$-type agents, is there a symmetric voting rule that can generate these interim allocation probabilities?}

We completely characterize these interim allocation probabilities. We call them {\sl reduced form implementable}
symmetric voting rules. The reduced form implementable symmetric voting rules are characterized by a family of $2(n+1)$ linear inequalities, where
$n$ is the number of agents. The extreme points of these symmetric voting rules are (i) a family of $(n+1)$ {\sl qualified majority}
voting rules and (ii) a family of $(n+1)$ {\sl qualified anti-majority} voting rules. A qualified majority (anti-majority) voting
rule is characterized by a quota $K$, and chooses alternative $a$ (respectively, $b$) whenever at least
$K$ agents vote for $a$. As a corollary, we show that every symmetric voting rule is {\sl reduced form
equivalent} (i.e., generating the same interim allocation probabilities)
to a convex combination of qualified majority and qualified anti-majority voting rules.
Both these families contain only deterministic voting rules.

We extend our characterization to monotone voting rules, i.e.,
voting rules that select $a$ with higher probability as the number of $a$-types increase. 
Monotone voting rules are strategy-proof (dominant strategy incentive compatible). The reduced form implementable symmetric monotone voting rules are characterized by a family of $(n+2)$ linear inequalities. The
extreme points of these rules are the family of $(n+1)$ qualified majority rules and a constant
rule that selects alternative $b$ at all type profiles. We use this result to 
show that an ex-ante Rawlsian rule (that maximizes the minimum of expected
utility of $a$-type agents and $b$-type agents) is a convex combination 
of a pair of qualified majority rules. We also investigate the reduced form question under a weaker notion 
of incentive constraints: {\sl ordinal Bayesian incentive
compatibility (OBIC)}~\citep{DP88,MS04,M16}. We show its connection 
to reduced form implementation by monotone voting rules.

We extend our characterizations for unanimous symmetric voting rules: a voting rule is unanimous if it chooses $a$ ($b$) whenever all the agents have type $a$ (respectively, $b$).
Using this, we characterize
the symmetric priors for which OBIC is implied by symmetry and unanimity. For independent priors,
this is the case when the probability of $a$ type is sufficiently small or sufficiently high. If we allow for correlation (still maintaining
symmetry), the set of priors where symmetry and unanimity implies OBIC contains priors where extreme
type profiles with low and high number of $a$ types are chosen with high probability.

We believe our results will be useful in designing optimal mechanisms in various models
of voting over a pair of alternatives. Indeed, Border's theorem is
extensively used in auction theory and mechanism design: for designing optimal auctions with budget constrained bidders
\citep{PV14}; for designing optimal verification mechanisms \citep{BDL14,MZ17,L20,L21}; for designing
symmetric auctions \citep{DP17}, and so on. The advantage of using a reduced form
in mechanism design problems is that they are in lower dimensional spaces than
the ex-post allocation problems. For instance, in the problem we study, the reduced form is two dimensional
but the (ex-post) voting rules are $n$-dimensional, where $n$ is the number of agents.
Our easy derivation of the ex-ante Rawlsian rule 
illustrates this advantage.

We give a detailed review of the literature in Section \ref{sec:lit}, but relate our results to Border's theorem here.
Consider Border's single object allocation problem but where each agent has two types (possible values
for the object): $\{0,1\}$. This is analogous to our problem where there are two types: $a$-type and $b$-type.
However, the voting problem in the current paper is a public good problem: the probability of choosing
$a$ and $b$ is the same across all the agents. The single object allocation problem is a private good problem
where the probability of choosing $a$ and $b$ may differ across agents. This makes the feasibility constraints
of allocation rules different in both the problems.

\cite{GK22} use a geometric approach (using support functions of convex sets) to study implementation in  social choice problems.  Their abstract formulation captures our problem too, and 
their results can be used to describe the support functions of our reduced form voting rules.
But, this neither describes the extreme points nor the necessary and sufficient
conditions that characterize the reduced form voting rules. \footnote{Further, they assume independent priors which we do not assume. They use their support function characterization to rederive Border's result.}
Indeed, it is not clear that an analogue of Border's theorem can exist in the voting problem.
In an important paper, \citet{GNR18} show that in a simple public good model with 
two alternatives, no computationally tractable characterization of reduced form allocation rules
is possible. Though this negative result applies to our model,  they allow reduced form implementation via asymmetric mechanisms.
By only looking at symmetric mechanisms, we overcome this impossibility: our characterization
admits a computationally tractable description of reduced form probabilities by a system of
(linear in number of voters) linear inequalities.

The rest of the paper is organized as follows. Section \ref{sec:model} introduces the model.
 Section \ref{sec:red} provides the main result of the paper: a characterization of the reduced form
 implementable voting rules. Section \ref{sec:mon} extends the main result by requiring monotone
 implementation, and provides an application to finding a Rawlsian voting rule. 
 Section \ref{sec:unan} extends the main characterization with unanimity
 and Section \ref{sec:large} for large economies.
Section \ref{sec:lit} gives a detailed literature review. The missing proofs are in Appendix \ref{sec:app1}. Proofs of Theorem \ref{theo:unan} and Theorem \ref{theo:uext} are similar to Theorem \ref{theo:main} and Theorem \ref{theo:ext} respectively.
 So, they have been provided in a separate appendix (Appendix \ref{sec:supapp}).

\section{The model}
\label{sec:model}

Let $N=\{1,\ldots,n\}$ be a finite set of agents (voters), where $n \ge 2$. Let $A=\{a,b\}$ be the set of two social alternatives (for instance,
a status-quo and a new alternative).
Each agent has a strict ranking of $A$. Hence, the preference of an agent can be expressed by her top ranked alternative.
We call it the {\sl type} of the agent. The type of agent $i$ is denoted as $t_i \in \{a,b\}$, which means that $t_i$
is the top ranked alternative of agent $i$. Hence, the set of all types (type space) is $A$ and the set of
all type profiles is $A^n$. A type profile in $A^n$ is denoted by $t \equiv (t_1,\ldots,t_n)$.

{\sc Exchangeable Prior.} Let $G$ be a probability distribution over type profiles. We assume
$G$ to be {\sl exchangeable}, i.e., for every type profile $t$ and every permutation $\sigma$,
$G(t)=G(t^{\sigma})$, where $t^{\sigma}$ is the permuted type profile. In this sense,
the probability of a type profile is only a function of number of agents having type $a$.
So, for every $k \in \{0,\ldots,n\}$, for any set of $k$ agents, the probability that exactly these
agents have type $a$ (and other agents have type $b$) is given by $\lambda(k)$.
By exchangeability, the probability a type profile has exactly $k$ agents of type $a$
is $C(n,k)\lambda(k)$, where $C(n,k)$ denotes the number of $k$-combinations from a set of $n$ elements.

We denote the marginal probability of any agent having type $a$ as $\pi$ and having
type $b$ as $(1-\pi)$.

{\sc Voting rule.} A {\sl voting rule} is a map $q: A^n \rightarrow [0,1]$, where $q(t)$ denotes the
probability with which alternative $a$ is chosen (and, hence, $1-q(t)$ is the probability with which
alternative $b$ is chosen) at type profile $t$. We will only consider {\sl symmetric or anonymous} voting rules,
i.e., for any permutation $\sigma$, we will require $q(t) = q(t^{\sigma})$ for all $t \in A^n$, where $t^{\sigma}$
is type profile obtained by permuting $t$ using the permutation $\sigma$. With a
slight abuse of notation, we will write $q$ as a map $q: \{0,1,\ldots,n\} \rightarrow [0,1]$, i.e., $q(k) \in [0,1]$
denotes the probability with which alternative $a$ is chosen at any type profile with $k$ votes for $a$.\footnote{We restrict ourselves to ordinal voting rules. Any cardinal voting rule
in a two alternative model must be ordinal if it is incentive compatible~\citep{MS04}.
Since reduced forms are usually used along with incentive constraints, restricting
attention to ordinal voting rule is without loss of generality in this sense. Even
without incentive constraints, \citet{ST12,AK14} show that restricting attention
to ordinal voting rules is without loss of generality if the planner is optimizing over
interim utilities of agents.} We only discuss symmetric voting rules, and whenever we 
refer to a voting rule from now on, we will mean a symmetric voting rule.

Given a voting rule $q$, we can compute the interim probability of each alternative being chosen.
If an agent has type $a$, the probability that alternative $a$ is chosen by voting rule $q$ is
denoted by $Q(a)$. To relate $Q$ and $q$, denote the probability that there are $k$ agents of type $a$ as
\begin{align*}
B(k) &:= \lambda(k) C(n,k) \qquad~\forall~k \in \{0,\ldots,n\}
\end{align*}
Note the following:
\begin{align*}
\sum_{k=0}^n B(k) &= 1~~~~\textrm{and}~~~\sum_{k=0}^n k B(k) = n\pi 
\end{align*}
The second equality follows because both $n\pi$ and $\sum_k k B(k)$ denote  the expected number of agents who have type $a$.

Using this, $Q$ can be computed from $q$ as follows.
\begin{align*}
n\pi Q(a) &= \sum_{k=0}^{n}k q(k)B(k),
\end{align*}
where both the LHS and the RHS computes the expected number of $a$-types who get $a$.
Hence,
\begin{align*}
Q(a) &= \frac{1}{n\pi}\sum_{k=0}^{n}k q(k)B(k),
\end{align*}
Similarly, if an agent has type $b$, the probability that alternative $a$ is chosen by voting rule $q$ is
\begin{align*}
Q(b) &= \frac{1}{n(1-\pi)} \sum_{k=0}^{n}(n-k)q(k)B(k)
\end{align*}
Of course, $1-Q(a)$ and $1-Q(b)$ denote the interim probabilities with which alternative $b$ is chosen
for types $a$ and $b$ respectively.

\section{Reduced form implementation}
\label{sec:red}

The interim allocation probabilities are two dimensional. Hence, they are easy to work with.
Some interim allocation probabilities are clearly not possible: for instance $Q(a)=1,Q(b)=0$
is impossible for $n\geq 2$ because any voting rule for which $Q(a)=1$ must choose $a$ at some profiles
where other agents have type $b$. By symmetry, $Q(b) \ne 0$. Then, the reduced form
question is what interim allocation probabilities are possible.

\begin{defn}
Interim allocation probabilities $Q \equiv (Q(a),Q(b)) \in [0,1]^2$ is {\bf reduced form implementable} if there exists a voting rule $q$ such that
\begin{align*}
\frac{1}{n\pi} \sum_{k=0}^{n}k q(k) B(k) &= Q(a) \\
\frac{1}{n(1-\pi)} \sum_{k=0}^{n}(n-k) q(k)B(k) &= Q(b) \\
0 \le q(k) &\le 1 \qquad~\forall~k \in \{0,\ldots,n\}
\end{align*}
\end{defn}

To see what kind of conditions are necessary for reduced form implementation, consider the following 
setting. Suppose there is a cost $j \in \{0,1,\ldots,n\}$
of choosing alternative $a$ but alternative $b$ costs zero. For any $a$-type agent, suppose the value
of alternative $a$ is $1$ and that of alternative $b$ is $0$. The expected value of $a$-types minus the cost
of choosing an alternative from a voting rule $q$ is
\begin{align}\label{eq:optn}
\sum_{k=0}^n (k-j)q(k)B(k) &= \frac{1}{n} \Big[(n-j)\sum_{k=0}^n kq(k)B(k) -j \sum_{k=0}^n(n-k) q(k) B(k) \Big] \\
&= (n-j)\pi Q(a) - j (1-\pi)Q(b) \nonumber
\end{align}

The LHS of (\ref{eq:optn}) is maximized by setting $q(k)=0$ if $k < j$ and $q(k)=1$ if $k \ge j$.
Hence, an upper bound for LHS of (\ref{eq:optn}) is $\sum \limits_{k=j}^n (k-j)B(k)$.
Similarly, the LHS of (\ref{eq:optn}) is minimized by setting
$q(k)=1$ if $k < j$ and $q(k)=0$ if $k \ge j$. Hence, a lower bound for LHS of (\ref{eq:optn}) is
$\sum \limits_{k=0}^j (k-j)B(k)$.
Thus, for any $j \in \{0,1,\ldots,n\}$,
\begin{align}
\sum_{k=j}^n (k-j)B(k) &\ge (n-j)\pi Q(a) - j (1-\pi)Q(b) \ge \sum_{k=0}^j (k-j)B(k) \label{eq:lub}
\end{align}
So, the inequalities (\ref{eq:lub}) are necessary for reduced form implementation.
Our main result says they are sufficient.

\begin{theorem}
\label{theo:main}
Interim allocation probabilities $Q$ is reduced form implementable if and only if
\begin{align}
j (1-\pi)Q(b) - (n-j)\pi Q(a) + \sum_{k=j}^n (k-j) B(k) &\ge 0 \qquad~\forall~j \in \{0,\ldots,n\} \label{eq:e1} \\
(n-j) \pi Q(a) - j (1-\pi)Q(b) + \sum_{k=0}^j (j-k) B(k) &\ge 0 \qquad~\forall~j \in \{0,\ldots,n\} \label{eq:e2}
\end{align}
\end{theorem}

The sufficiency part of proof of Theorem \ref{theo:main} and other results are in Appendix \ref{sec:app1}.
It is proved by first describing the extreme points of all reduced form implementable voting rules (Theorem \ref{theo:ext})
and then showing that the extreme points of the system (\ref{eq:e1}) and (\ref{eq:e2}) correspond
to exactly the same voting rules.

The reduced form implementable voting rules are described by $2(n+1)$ inequalities. Out of
this, four correspond to non-negativity of $Q(a),Q(b)$ and upper bounding of
$Q(a),Q(b)$ by $1$. The rest of the $2(n-1)$ inequalities restrict the space
of interim allocation probabilities in the unit square. To see this, consider the uniform prior (independent
prior) with $\pi=\frac{1}{2}$ and $n=3$. In this case, $(Q(a),Q(b))$ is reduced form implementable if and only if
\begin{align*}
2Q(a) - Q(b) &\le \frac{5}{4},~~~Q(a) - 2Q(b) \le \frac{1}{4},~~~Q(b) - 2Q(a) \le \frac{1}{4},~~~2Q(b) - Q(a) \le \frac{5}{4} \\
Q(a),Q(b) &\in [0,1]
\end{align*}

\begin{figure}
\centering
\includegraphics[width=3in]{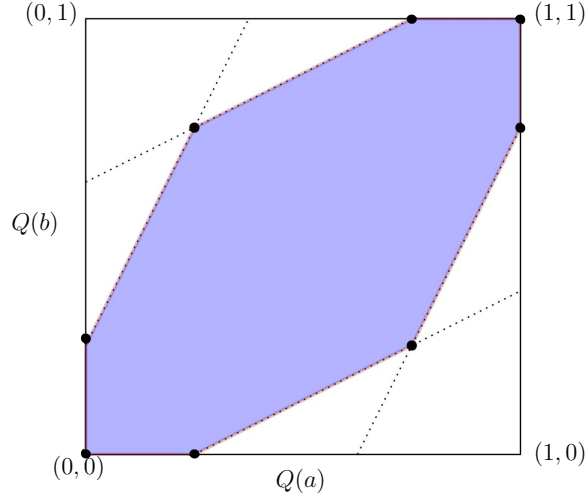}
\caption{Polytope of reduced form implementable voting rules}
\label{fig:main}
\end{figure}

The polytope enclosed by these inequalities is shown in Figure \ref{fig:main}. One sees 8 extreme
points of this polytope, two of them correspond to the constant allocation rules ($(0,0)$ correspond
to $b$ always chosen and $(1,1)$ correspond to $a$ always chosen). The rest of them
belong to a family of voting rules which we call qualified majority and qualified anti-majority.
We establish this result next. This allows us to show that any reduced form implementable
voting rule is ``equivalent" to a convex combination of voting rules from this set.

\begin{defn}
Two voting rules $q$ and $\hat{q}$ are {\bf reduced form equivalent} if they generate
the same interim allocation probabilities: $Q(a)=\widehat{Q}(a)$ and $Q(b)=\widehat{Q}(b)$.
\end{defn}

We now introduce two classes of voting rules which will be useful to describe the extreme points
of reduced form implementable voting rules.
\begin{defn}
A voting rule $q^+$ is a {\bf qualified majority} if there exists $j\in\{0,\ldots,n\}$ such that for all $k \in \{0,\ldots,n\}$
\begin{align*}
q^+(k) =
\begin{cases}
1 & \textrm{if}~k \ge j \\
0 & \textrm{otherwise}
\end{cases}
\end{align*}
We call such a voting rule a qualified majority with quota $j$.

 A voting rule $q^-$ is {\bf qualified anti-majority} if there exists $j\in\{0,\ldots,n\}$ such that for all $k \in \{0,\ldots,n\}$
\begin{align*}
q^-(k) =
\begin{cases}
1 & \textrm{if}~k < j \\
0 & \textrm{otherwise}
\end{cases}
\end{align*}
We call such a voting rule a qualified anti-majority with quota $j$.
\end{defn}
The definition of qualified majority is similar to \citet{AK14}. The only
difference is that if
quota is $j$, they allow $q^+(j)$ to take any value in $[0,1]$, but we break the
tie deterministically.

If $q^j$ is a qualified majority with quota $j$, then its reduced form probabilities are
\begin{align*}
Q^j(a) &= \frac{1}{n\pi} \sum_{k=0}^n k q^j(k) B(k) = \frac{1}{n\pi} \sum_{k=j}^n k B(k) \\
Q^j(b) &= \frac{1}{n(1-\pi)} \sum_{k=0}^n (n-k) q^j(k) B(k) = \frac{1}{n(1-\pi)} \sum_{k=j}^{n} (n-k) B(k) \\
\end{align*}
 Notice that when $j=0$, we have $Q^0(a)=Q^0(b)=1$. This corresponds to the constant voting rule
where $a$ is chosen at every type profile.

If $\bar{q}^j$ is a qualified anti-majority with quota $j$, then its reduced form probabilities are
\begin{align*}
\overline{Q}^j(a) &= \frac{1}{n\pi} \sum_{k=0}^n k \bar{q}^j(k) B(k) = \frac{1}{n\pi}\sum_{k=0}^{j-1} k B(k) \\
\overline{Q}^j(b) &= \frac{1}{n(1-\pi)} \sum_{k=0}^n (n-k) \bar{q}^j(k) B(k) = \frac{1}{n(1-\pi)} \sum_{k=0}^{j-1} (n-k) B(k) \\
\end{align*}

Denote the set of all qualified majority voting rules by $\mathcal{Q}^+$ and the set of
all qualified anti-majority voting rules by $\mathcal{Q}^-$. Notice that when $j=0$, we have
$\overline{Q}^0(a)=\overline{Q}^0(b)=0$. This corresponds to the constant voting rule where $b$ is chosen at every
type profile. Hence, $\mathcal{Q}^+ \cup \mathcal{Q}^-$ contains the two constant voting rules.

\begin{theorem}
\label{theo:ext}
Every symmetric voting rule is reduced-form equivalent to a convex combination of voting rules
in $\mathcal{Q}^+ \cup \mathcal{Q}^-$.
\end{theorem}

We compare our results to some of the results in \citet{AK14}. They consider
a cardinal voting model with two alternatives, where type of an agent (a one-dimensional number
with finite support) gives cardinal utilities of two alternatives. They consider cardinal voting 
rules and Bayesian incentive compatibility (BIC). They have two main results with symmetric cardinal
voting rules: (a) a utilitarian maximizer in the class of BIC and symmetric rules is a qualified majority; (b) an 
interim efficient, BIC and symmetric rule is a qualified majority.\footnote{They have analogues
of these results without symmetry too. A {\sl weighted majority} rule is interim efficient and BIC. 
Similarly, a weighted majority rule is utilitarian maximizer in the class of BIC rules.}

While related, their results and our results are not comparable. First, we only consider ordinal
voting rules, while they allow for cardinal rules. Second, the types of agents in their 
model are independent while we allow for correlated types -- exchangeable distributions allow for
correlation.

Third, Theorem \ref{theo:ext} says 
that the extreme points of the set of reduced form implementable voting rules consist 
of qualified majority and qualified anti-majority rules. We do not require incentive compatibility or any
additional axiom (like interim efficiency) for this result. In the next section, we will impose
monotonicity (equivalent to dominant strategy incentive compatibility) of voting rules, 
and show that the the extreme points of the set of monotone reduced form implementable voting rules consist
of qualified majority rules and a constant rule. As we discuss in Section \ref{sec:raw}, 
our results are useful in settings where the objective function of the planner is not linear.

Finally, 
we explore the consequences of imposing unanimity on the reduced form implementation in Section
\ref{sec:unan}. Unanimity is a much weaker axiom than interim efficiency used in \citet{AK14}. Theorem \ref{theo:uext}
describes the extreme points of reduced form implementable rules satisfying unanimity and this contains 
rules that are {\sl not} qualified majority.

\section{Monotone reduced form implementation}
\label{sec:mon}

A natural restriction on voting rules is monotonicity. Formally, a symmetric voting rule $q$
is {\sl monotone} if $q(k) \ge q(k-1)$ for all $k \in \{1,\ldots,n\}$.
Monotonicity is equivalent to {\sl strategy-proofness} or {\sl dominant strategy incentive
compatibility} in voting models with
two alternatives.

\begin{defn}
Interim allocation probabilities $Q \equiv (Q(a),Q(b)) \in [0,1]^2$ is {\bf reduced form monotone implementable} if there exists a monotone voting rule $q$ whose interim allocation probabilities equal $Q$.
\end{defn}

With the help of our main results, we can characterize the reduced form monotone implementable
interim allocation probabilities.

\begin{theorem}
\label{theo:monimp}
Let $Q \equiv (Q(a), Q(b))$ be any interim allocation probabilities. Then, the following statements are equivalent.

\begin{enumerate}

\item $Q$ is reduced form monotone implementable.

\item $Q$ is reduced form implementable and $Q(a) \ge Q(b)$.

\item $Q$ is reduced form implementable by convex combination of
qualified majority voting rules and a constant voting rule that selects $b$ at all type profiles.

\item $Q$ satisfies
\begin{align}
j (1-\pi)Q(b) - (n-j)\pi Q(a) + \sum_{k=j}^n (k-j) B(k) &\ge 0 \qquad~\forall~j \in \{0,\ldots,n \} \label{eq:mm1} \\
Q(a) - Q(b) &\ge 0 \label{eq:mm2}
\end{align}

\end{enumerate}
\end{theorem}

We make two remarks about Theorem \ref{theo:monimp}. 

\begin{itemize}

\item {\bf Equivalence of notions of IC under independent priors.} Note that Theorem \ref{theo:monimp} holds for correlated (exchangeable) priors. The equivalence of (1) and (2) in Theorem \ref{theo:monimp} is 
related to equivalence of strategy-proof and Bayesian incentive compatibility 
in some mechanism design models with independent priors~\citep{MV10,GKMS13}. To understand
this better, consider a natural notion of Bayesian incentive compatibility in ordinal mechanisms.
{\sl Ordinal Bayesian incentive compatibility (OBIC)} requires that the
truthtelling lottery first-order stochastically dominates any lottery that can be obtained
by a misreport \citep{DP88,MS04,M16}.

Formally, fix a voting rule $q$. 
Let $Q(x|y)$ denote
the interim probability of getting $a$ by reporting $x$ in the voting rule when true type is $y$. So, for an $a$-type 
agent with utilities $u(a)$ and $u(b)$ for $a$ and $b$ respectively (with $u(a) > u(b)$ since the agent is $a$-type), the IC constraint is
\begin{align*}
u(a) Q(a|a) + u(b) (1-Q(a|a)) & \ge u(a) Q(b|a) + u(b) (1-Q(b|a)) \\
\Leftrightarrow (u(a)-u(b)) Q(a|a) &\ge (u(a)-u(b)) Q(b|a) \\
\Leftrightarrow Q(a|a) &\ge Q(b|a)
\end{align*}
where the last equivalent inequality follows because $u(a) > u(b)$.
Similarly, the IC constraint for $b$-type is $1-Q(b|b) \ge 1-Q(a|b)$ or $Q(a|b) \ge Q(b|b)$.

If prior is independent,
then $Q(x|y)=Q(x)$. Then, OBIC is equivalent to requiring $Q(a) \ge Q(b)$.
This is the constraint in (2) and (4) of Theorem \ref{theo:monimp}.
Hence, by Theorem \ref{theo:monimp}, we have the following corollary.

\begin{cor}
\label{cor:equiv}
Suppose the prior  is independent and $Q \equiv (Q(a), Q(b))$ be any interim allocation probabilities.
Then, each of (1) to (4) in Theorem \ref{theo:monimp} is equivalent to the following statement
\begin{itemize}
\item $Q$ is reduced form implementable by an OBIC voting rule.
\end{itemize}
\end{cor}

By the equivalence of (1) and (2) in Theorem \ref{theo:monimp}, Corollary \ref{cor:equiv} implies that every OBIC voting rule is reduced-form
equivalent to a strategy-proof voting rule under independent priors. This OBIC and strategy-proof equivalence result is a corollary of an important (and more general) result on equivalence of strategy-proof and Bayesian incentive compatible mechanism with independent types in \citet{GKMS13}.
Corollary \ref{cor:equiv} describes the reduced form inequalities that characterize OBIC voting rules
with independent priors and shows that they are the same reduced form inequalities 
that describe monotone voting rules.

In voting models with at least three alternatives, {\sl ex-post} equivalence of deterministic strategy-proof and OBIC voting 
rules is established
for generic independent priors in \citet{MS04} and \citet{M16} under unanimity constraints. 

\item {\bf Extreme points of voting rules.} A voting rule $q$ is {\sl extreme}
if there does not exist a pair of voting rules $\bar{q}$ and $\tilde{q}$ such that for some $\lambda \in (0,1)$,
$q(k) = \lambda \bar{q}(k) + (1-\lambda) \tilde{q}(k)$ for all $k$.
Let $\mathcal{Q}^{ex}$ be the set of all extreme voting rules.

A voting rule $q$ is {\sl reduced-form extreme}
if there does not exist a pair of voting rules $\bar{q}$ and $\tilde{q}$ with interim allocation 
probabilities $\bar{Q}$ and $\tilde{Q}$ respectively, such that for some $\lambda \in (0,1)$,
$Q(x) = \lambda \bar{Q}(x) + (1-\lambda) \tilde{Q}(x)$ for all $x \in \{a,b\}$. Let $\mathcal{Q}^{rex}$
be the set of all reduced-form extreme voting rules. By Theorem \ref{theo:ext},
$\mathcal{Q}^{rex} = \mathcal{Q}^+ \cup \mathcal{Q}^-$.

It is easy to see that {\sl every} deterministic voting rule is an extreme voting rule, i.e., 
belongs to $\mathcal{Q}^{ex}$.
For instance, suppose $n=4$, a voting rule that chooses $b$ if there are exactly two
$a$-types and chooses $a$ otherwise belongs to $\mathcal{Q}^{ex}$. However
this voting rule is neither a qualified majority nor a qualified anti-majority. Hence,
it does not belong to $\mathcal{Q}^{rex}$, and hence, we have $\mathcal{Q}^{rex} \subsetneq \mathcal{Q}^{ex}$.
That is, the set of extreme points of voting rules in the reduced form is a strict subset
of the set of extreme points of voting rules in the ex-post form. This difference
disappears once we impose monotonicity.

To see this, let $\mathcal{Q}^{mex}$ denote the set of monotone extreme voting rules
and $\mathcal{Q}^{mrex}$ denote the set of monotone reduced-form extreme voting rules.
By Theorem \ref{theo:monimp}, $\mathcal{Q}^{mrex}$ consists of qualified majority 
voting rules and the constant voting rule that selects $b$ at all type profiles.
\citet{PS12} show that $\mathcal{Q}^{mex}$ consists of the same set of voting rules.\footnote{To be precise,
\citet{PS12} do not restrict attention to symmetric voting rules and characterize
the extreme points of all monotone voting rules as the set of {\sl voting by committee} rules
introduced in \citet{BSZ91}. Imposing
symmetry gives us the required set of symmetric monotone extreme voting rules.} Hence,
we can conclude that $\mathcal{Q}^{mex}=\mathcal{Q}^{mrex}$.

\end{itemize}

\subsection{Application: Rawlsian rule}
\label{sec:raw}

In this section, we apply Theorem \ref{theo:monimp} to characterize an {\sl ex-ante Rawlsian rule}. We say an agent is ``satisfied" if its
top ranked alternative is chosen.  An ex-ante Rawlsian rule maximizes the minimum number
of satisfied agents between $a$-types and $b$-types over all monotone voting rules. Formally, fix any voting rule $q$.  The expected number
of $a$-type satisfied agents  is 
$$\sum_{k=0}^n k q(k) B(k) = n\pi Q(a)$$
Similarly,  the expected number of 
$b$-type satisfied agents is 
$$\sum_{k=0}^n (n-k) (1-q(k)) B(k) = n(1-\pi) (1-Q(b))$$
An {\sl ex-ante Rawlsian}   rule maximizes the minimum number
of satisfied agents between $a$-types and $b$-types.

\begin{defn}
A monotone voting rule $q^R$ is {\bf ex-ante Rawlsian} if for every monotone voting rule $q$,
\begin{align*}
\min \big(\pi Q^R(a), (1-\pi)(1-Q^R(b)\big) &\ge \min \big(\pi Q(a), (1-\pi) (1-Q(b)\big)
\end{align*}
\end{defn}

Using Theorem \ref{theo:monimp}, we provide a complete description of
the ex-ante Rawlsian rule: it is a convex combination of a pair of qualified majority
voting rules.
\begin{prop}
\label{prop:rawls}
The ex-ante Rawlsian rule $q^R$ is a convex combination of qualified majority
with quota $j^*$ and $(j^*+1)$, where
\begin{align}
j^*=  \max \{j\in\{0,\cdots,n\}: \sum_{k=j}^n  B(k) \geq 1-\pi\}
\end{align}
The interim allocation probabilities corresponding to $q^R$ are
\begin{align}
Q^R(a)&=\frac{1}{n\pi} \left(j^*(1-\pi)+ \sum_{k=j^*}^n (k-j^*) B(k)\right)\\
Q^R(b)&=\frac{1}{n(1-\pi)} \left((n-j^*)(1-\pi)- \sum_{k=j^*}^n (k-j^*)  B(k)\right)
\end{align}
\end{prop}

The optimal quota $j^*$ is determined by comparing the joint probability that at least $j^*$ agents is $a$-type and the marginal probability of $b$-type (which is $1-\pi$).  For qualified majority with quotas $j^*$ and $j^*+1$, the joint probability that at least $j^*$ agents is $a$-type is approximately equal to the ex ante probability that alternative $a$ is chosen from these rules. Then optimal quota $j^*$ is selected such that  the ex ante probability that alternative $a$ is chosen  is approximately equal to the marginal probability of $b$-type.

\section{Unanimity constraints}
\label{sec:unan}

We now impose a familiar axiom on the voting rule.
A voting rule $q$ is {\sl unanimous} if $q(n)=1$ and $q(0)=0$.   Unanimity imposes restrictions on the interim allocation probabilities. For instance, consider a unanimous
voting rule $q$. Then, its interim allocation probabilities must be
\begin{align*}
Q(a) &= \frac{1}{n\pi} \sum_{k=0}^{n}k q(k) B(k) = \frac{1}{n\pi} \big[\sum_{k=1}^{n-1}k q(k) B(k) + n B(n) \big]\\
Q(b) &= \frac{1}{n(1-\pi)} \sum_{k=1}^{n-1}(n-k) q(k)B(k)
\end{align*}
Hence, the reduced-form characterization changes as in the theorem below.
\begin{defn}
Interim allocation probabilities $Q(a),Q(b) \in [0,1]$ is {\bf reduced form unanimous (u-)implementable} if there exists a
unanimous voting rule $q$ whose interim allocation probabilities equal $Q$.
\end{defn}
Notice that $q$ is $(n-2)$-dimensional since the values of $q(0)$ and $q(n)$ are fixed.

\begin{theorem}
\label{theo:unan}
Interim allocation probabilities $Q$ is reduced form u-implementable if and only if
\begin{align}
j (1-\pi)Q(b) - (n-j)\pi Q(a) + \sum_{k=j}^n (k-j) B(k) &\ge 0 \qquad~\forall~j \in \{0,\ldots,n \} \label{eq:u1} \\
(n-j) \pi Q(a) - j (1-\pi)Q(b) + \sum_{k=0}^j (j-k) B(k) &\ge j \lambda(0) + (n-j)\lambda(n) \qquad~\forall~j \in \{0,\ldots,n \} \label{eq:u2}
\end{align}
\end{theorem}
The proofs of Theorem \ref{theo:unan} and Theorem \ref{theo:uext} are in Appendix \ref{sec:supapp}.
They are similar to Theorem \ref{theo:main} and Theorem \ref{theo:ext}.

For $n=3$ and the uniform prior with $\pi=\frac{1}{2}$, the set of reduced form u-implementable voting rules
are shown in the smaller polytope in Figure \ref{fig:unan}. It lies inside the polytope characterizing
the set of all reduced form implementable voting rules. This polytope has only four extreme points.
We characterize them next.

\begin{figure}
\centering
\includegraphics[width=5in]{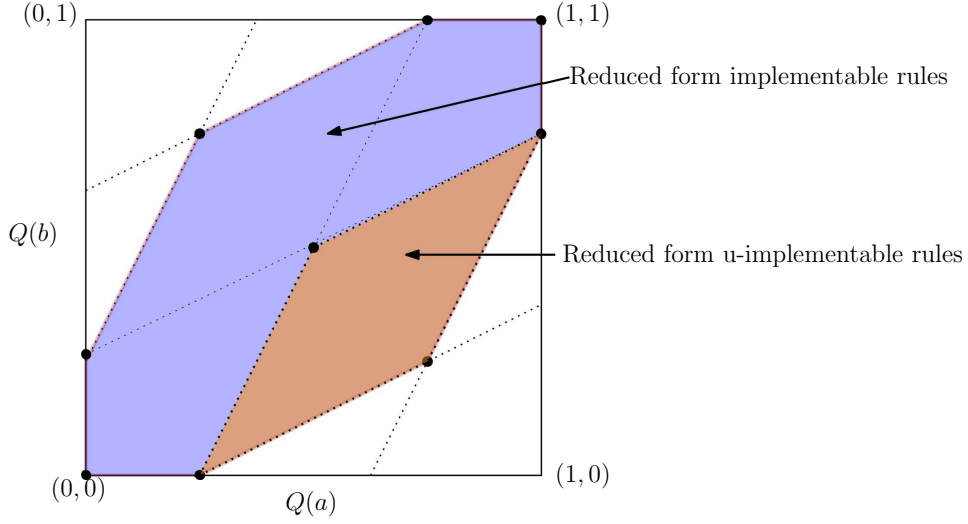}
\caption{Polytope of reduced form u-implementable voting rules}
\label{fig:unan}
\end{figure}

The extreme points of reduced-form implementable unanimous voting rules are
defined by two new families of unanimous voting rules.

\begin{defn}
A  voting rule $q^+_u$ is {\bf u-qualified majority} if it is a qualified majority with quota $j$, where
$j \in \{1,\ldots,n\}$. We call such a voting rule a u-qualified majority with quota $j$.

A voting rule $q^-_u$ is {\bf u-qualified anti-majority} if there exists $j \in \{1,\ldots,n\}$
\begin{align*}
q^-_u(k) =
\begin{cases}
1 & \textrm{if}~k \in \{1,\ldots,j-1\} \cup \{n\} \\
0 & \textrm{otherwise}
\end{cases}
\end{align*}
We call such a voting rule a u-qualified anti-majority with quota $j$.
\end{defn}

A u-qualified majority is just a non-constant qualified majority rule. On the other hand,
a u-qualified anti-majority is {\sl not} merely a non-constant qualified anti-majority.
A u-qualified anti-majority is constructed by taking a non-constant qualified anti-majority
and making it unanimous. For instance if $n=4$ and quota $j=2$,
a qualified anti-majority will set $q(0)=q(1)=1, q(2)=q(3)=q(4)=0$.
But a u-qualified anti-majority will set $q(0)=0, q(1)=1, q(2)=q(3)=0, q(4)=1$.

We write down the interim allocation probabilities of a u-qualified majority and u-qualified
anti-majority below. If $q^+_u$ is a u-qualified majority with quota $j$, then
\begin{align*}
Q^+_u(a) &= \frac{1}{n\pi} \sum_{k=j}^n k B(k) \\
Q^+_u(b) &= \frac{1}{n(1-\pi)} \sum_{k=j}^n (n-k)B(k)
\end{align*}
On the other hand, if $q^-_u$ is a u-qualified anti-majority with quota $j$, then
\begin{align*}
Q^-_u(a) &= \frac{1}{n\pi} \Big[\sum_{k=1}^{j-1} k B(k) + n B(n) \Big]\\
Q^-_u(b) &= \frac{1}{n(1-\pi)} \sum_{k=1}^{j-1} (n-k)B(k)
\end{align*}

Denote the set of all u-qualified majority voting rules by $\mathcal{Q}_u^+$ and the set of
all u-qualified anti-majority voting rules by $\mathcal{Q}_u^-$.
Notice that the u-qualified majority with quota $n$ and the
u-qualified anti-majority with quota $1$ are the same voting rules. Similarly, the u-qualified
majority with quota $1$ and the u-qualified anti-majority with quota $n$ are the same voting rules.
Hence, these two families of voting rules contain a total of $2(n-1)$ unanimous
voting rules.
The following theorem shows that they form the extreme points of all reduced form u-implementable
voting rules.

\begin{theorem}
\label{theo:uext}
Every symmetric and unanimous voting rule is reduced-form equivalent to a convex combination of
voting rules in $\mathcal{Q}_u^+ \cup \mathcal{Q}_u^-$.
\end{theorem}

\subsection{When are incentive constraints implied?}
\label{sec:icfree}

Corollary \ref{cor:equiv} (and \citet{GKMS13}) shows that for independent priors,  every OBIC voting rule is reduced form equivalent to a strategy-proof voting rule. This reduced form equivalence, however, fails with unanimity constraint, i.e., not every OBIC and unanimous voting rule is reduced form equivalent to a strategy-proof and unanimous voting rule. The following example presents an OBIC and unanimous voting rule that is not reduced form equivalent to a strategy-proof and unanimous voting rule.

\begin{example}\label{ex:obiuna}
\end{example}
Suppose $n=3$ and the prior is independent with  $\pi=\frac{1}{2}$: so, $B(0)=\frac{1}{8}, B(1)=\frac{3}{8}, B(2)=\frac{3}{8}, B(3)=\frac{1}{8}$. Consider $Q(a)=Q(b)= \frac{1}{2}$.  Then $Q$ is OBIC. We show that  $Q$ is implementable by a unique unanimous voting rule, but it is not strategy-proof. Let $q$ be any unanimous rule that implements $Q$. Then $q$ satisfies
\begin{align*}
Q(a) &=\frac{1}{3\pi}[ \sum_{k=1}^{2}kq(k)B(k)+3B(3)]=\frac{1}{4}\Big(q(1)+2q(2)+1\Big)= \frac{1}{2} \\
 Q(b) &=\frac{1}{3\pi}\sum_{k=1}^{2}(3-k)q(k)B(k)=\frac{1}{4}\Big(2q(1)+ q(2) \Big)= \frac{1}{2}
 \end{align*}
Hence $q(1)-q(2)=1$. Since $0\leq q(1),q(2)\leq 1$,  it implies that $q(1)=1$ and $q(2)=0$, i.e., $q$ is unique. However, $q$ is not strategy-proof.
\hfill $\blacksquare{}$

Imposing unanimity contracts the set of  reduced form implementable voting rules. In contrast to qualified anti-majority rules,  some u-qualified anti-majority rules can be OBIC.   The following result provides a necessary and sufficient condition on prior beliefs such that all unanimous voting rules are OBIC. 
\begin{prop}
\label{prop:icf}
Every unanimous and symmetric voting rule is  OBIC if and only if
\begin{align}
\label{eq:icf1}
\lambda(j) &\le \min\Big(\frac{\lambda(1)+ \lambda(n)}{C(n-1,j-1)}, \frac{\lambda(0)+ \lambda(n-1)}{C(n-1,j)} \Big)  \qquad~\forall~j \in \{1,\ldots,n-1\}
\end{align}
Further, if the prior is independent, every unanimous and symmetric voting rule is OBIC if and only if
 \begin{align}
 \label{eq:icf2}
 C(n-1,j-1) &\le  \Big[\Big(\frac{\pi}{1-\pi}\Big)^{n-j}+\Big(\frac{\pi}{1-\pi}\Big)^{1-j}\Big]~\qquad~\forall~j \in \{1,\ldots,n-1\}
\end{align}
\end{prop}

Using Corollary \ref{cor:equiv}, we can argue that when (\ref{eq:icf2}) holds and the prior is independent, every unanimous voting rule is reduced form equivalent to a strategy-proof voting rule.
An immediate corollary of the above result is that when there is a small number of agents, every unanimous voting rule is OBIC
if the prior is independent.

\begin{cor}If the prior is independent and $n=3$, every unanimous and symmetric voting rule is OBIC.
\end{cor}
\begin{proof}
Since $\pi \in (0,1)$, $j^* = \lfloor 3 \pi \rfloor \le 2$. If $j^*=1$, we get
\begin{align*}
B(1) = 3 \pi (1-\pi)^2 &\le 3 \pi (\pi^2 + (1-\pi)^2)
\end{align*}
If $j^*=2$, we get
\begin{align*}
B(2) = 3 \pi^2 (1-\pi) = \frac{3\pi}{2} (2\pi (1-\pi)) \le \frac{3\pi}{2} (\pi^2 + (1-\pi)^2)
\end{align*}
Hence, by Proposition \ref{prop:icf}, every unanimous voting rule is OBIC.
\end{proof}

To illustrate Proposition \ref{prop:icf}, suppose $n=4$.  The condition (\ref{eq:icf1}) is given by
\begin{align*}
 &3\lambda(2) \le \lambda(1)+\lambda(4)\\
 & 3\lambda(3) \le \lambda(1)+\lambda(4)\\
 & 3\lambda(1) \le \lambda(0)+\lambda(3)\\
 & 3\lambda(2) \le \lambda(0)+\lambda(3)
\end{align*}
 Notice that for independent uniform priors, $\lambda(k)=(\frac{1}{2})^4$, the belief conditions  fail. For sufficiently positively correlated beliefs where $\lambda(0)$ and $\lambda(4)$ are large, the belief conditions  hold. This is in general true. If $\lambda(0)$ and $\lambda(n)$ are sufficiently large,
 (\ref{eq:icf1}) holds. Similarly, if $\lambda(0)$ and $\lambda(1)$ (or, $\lambda(n-1)$ and $\lambda(n)$) are sufficiently large, (\ref{eq:icf1}) holds.

\section{Large Economies}
\label{sec:large}

In this section, we apply our results to large economies. For this, we
assume independent and identically distributed types. So, $\pi$ denotes
the probability that an agent is $a$-type.
Let $\mu:=n\pi$ denote the mean of the binomial distribution.

There are two ways in which we increase the value of $n$. First, we
fix the value of $\pi$ and increase $n$. This implies that the expected
number of $a$-types ($\mu$) also increases. Second, we fix the expected
number of $a$-types at $\mu$, and increase $n$. This implies that
the value of $\pi$ decreases with increasing $n$. We show the implication
of large $n$ on the set of reduced form implementable voting rules in both
the cases.

Since $n$ is variable in this section, for an arbitrary voting rule, we denote the interim allocation
probabilities as $(Q(a;n), Q(b;n))$ in this section. For a fixed $\pi$ and $n$,
the interim allocation probabilities corresponding to qualified majority and anti-qualified
majority voting rules will be useful for our analysis. In particular,
pick a qualified majority voting rule with quota $j > 0$.\footnote{Qualified majority
with quota $j=0$ corresponds to the constant voting rule where $a$ is chosen at
every type profile.}
For such a qualified majority, the interim allocation probabilities satisfy
\begin{align}
Q^j(a;n) - Q^j(b;n) &= \frac{1}{n\pi} \sum_{k=j}^n k B(k) - \frac{1}{n(1-\pi)} \sum_{k=j}^n (n-k) B(k) \nonumber \\
&= \frac{1}{n\pi} \sum_{k=j}^n k C(n,k) \pi^k (1-\pi)^{(n-k)} - \frac{1}{n(1-\pi)} \sum_{k=j}^n (n-k) C(n,k) \pi^k (1-\pi)^{(n-k)} \nonumber \\
&=  \sum_{k=j}^n C(n-1,k-1) \pi^{k-1} (1-\pi)^{(n-k)} - \sum_{k=j}^n C(n-1,k) \pi^k (1-\pi)^{(n-k-1)} \nonumber \\
&= C(n-1,j-1) \pi^{j-1} (1-\pi)^{(n-j)} \label{eq:diffq}
\end{align}
Similarly, for a qualified anti-majority with quota $j > 0$, the interim allocation probabilities satisfy
\begin{align}
\overline{Q}^j(b;n) - \overline{Q}^j(a;n) &= C(n-1,j-1) \pi^{j-1} (1-\pi)^{(n-j)} \label{eq:diffaq}
\end{align}
This can also be seen from the fact that for a fixed quota $j$, the qualified majority and
the qualified anti-majority interim allocation probabilities are related as: $\overline{Q}^j(a;n) = 1- Q^j(a;n)$ and $\overline{Q}^j(b;n) = 1- Q^j(b;n)$.

Depending on whether we increase $n$ for a fixed $\pi$ or fixed $\mu$, the RHS of (\ref{eq:diffq}) (and (\ref{eq:diffaq}))
behaves differently. In the former case, it is approximately equal
to a normal distribution with vanishing values of density. In the latter case, it is related to
the Poisson distribution. This leads to different convergence results in these cases.

\begin{prop}\label{prop:large1}
Suppose $\pi$ is fixed and $\pi \in (0,1)$. Then, for every $\epsilon > 0$, there exists $n_0$ such that for every $n$-agent economy  with $n > n_0$,  if interim allocation $(Q(a,n),Q(b,n))$ is reduced form implementable, then
\begin{align*}
|Q(a;n) - Q(b;n)|< \epsilon.
\end{align*}
\end{prop}

Proposition \ref{prop:large1} says that in large economies, the only reduced form implementable probabilities are those where $Q(a;n)=Q(b;n)$.\footnote{For correlated priors, it is well known that  the central limit theorem does not hold in general. However, we conjecture that Proposition  \ref{prop:large1} continues to hold for the case of infinite exchangeable priors, where we say an infinite sequence $ X_1, X_2, X_3,\dots$ of  random variables is exchangeable if for any finite $n$,  the joint probability distribution of  $(X_1,X_2,\dots,X_n)$ is the same as that of $(X_{\sigma(1)},X_{\sigma(2)},\dots,X_{\sigma(n)}) $ for any permutation $ \sigma$.}
 If the number of agents is large, the interim allocation probabilities (for any voting rule) is less sensitive
to the type of the agent. Hence, both $a$-types and $b$-types get the same interim allocation probabilities with large $n$.

However, this is not the case if the economies become
large with a fixed $\mu$. If $\mu$ is fixed, increasing $n$ decreases $\pi$. So, the probability of $a$-types decreases, i.e., $b$-types dominate the economy. As a result, depending
on how sensitive a voting rule is to the number of $b$-types (or $a$-types), we may get quite different interim allocation 
probabilities $Q(a;n)$ and $Q(b;n)$. For instance, consider the simple rule that chooses $b$ when all agents
have $b$-type and chooses $a$ otherwise. Then, if an agent has $a$-type, the rule must choose: $Q(a;n)=1$.
But if an agent has $b$-type, the rule chooses $b$ if all other $(n-1)$ agents have 
$b$ type. For a fixed $\mu$, the probability that a given agent has $b$ type is $1-(\mu/n)$.
So, probability that $(n-1)$ agents have $b$ type is $\big(1-(\mu/n)\big)^{n-1}$, which converges
to $e^{-\mu}$ for large $n$. So, for large $n$, we have $Q(b;n)=1-e^{-\mu}$, and $Q(a;n)-Q(b;n)=e^{-\mu} > 0$. The proposition below
uses a slightly more sophisticated voting rule to come up with an improved bound on $Q(a;n)-Q(b;n)$.

\begin{prop}\label{prop:large2}
Suppose $\mu$ is fixed. Then, there is a positive constant $M(\mu)$ such that for every $\epsilon > 0$, there exists $n_0$ such that for every $n$-agent economy with $n > n_0$,
\begin{enumerate}

\item interim allocation probabilities $(Q(a,n),Q(b,n))$ exists which is reduced form implementable and
\begin{align*}
Q(a;n) - Q(b;n) > M(\mu) - \epsilon
\end{align*}

\item interim allocation probabilities $(\widehat{Q}(a;n),\widehat{Q}(b;n))$ exists which is reduced form implementable and
\begin{align*}
\widehat{Q}(b;n) - \widehat{Q}(a;n) > M(\mu)  - \epsilon
\end{align*}

\end{enumerate}
\end{prop}

Combining Propositions \eqref{prop:large1}, \eqref{prop:large2} and Corollary \eqref{cor:equiv}, we conclude that  every reduced form implementable rule  is strategy-proof in the large for the fixed $\pi$, but this is not the case if $\mu$ is fixed.

\section{Relation to the literature}
\label{sec:lit}

The Border's theorem for single object allocation problem was formulated in \citet{M84,MR84}. The reduced form characterization
for this problem were developed in \citet{B91}. The symmetric version of Border's theorem with an elegant proof
using Farkas Lemma is developed in \citet{B07}. There are other approaches to proving Border's theorem (which also
makes it applicable in some constrained environment): network flow approach in \citet{CKM13}, geometric approach in
\citet{GK22}. \citet{HR15} provide an equivalence characterization of Border's theorem using second order
stochastic dominance. \citet{KMS21} further develop the majorization approach and apply it to a variety of problems
in economics.
 Border's theorem applies to private values single object auction, but
\citet{GK16} extend Border's theorem to allow for value interdependencies. \cite{Z21}  generalizes reduced-form characterizations to allocation of multiple objects with paramodular constraints.   \cite{LY21}  study a universal   implementation for allocation of multiple objects. \citet{Y21} considers the consequences of incorporating fairness constraints in the reduced form problem.  \citet{L22} considers a public good
allocation problem but with only two agents (but multiple alternatives). He provides an extension of Border's theorem to his two-agent
problem. Our ordinal voting model
over two alternatives is a public good model with a specific type space, which is not covered in these papers.

\cite{V11} studies the combinatorial structure of  reduced-form auctions by the polymatroid theory; see also \citet{CKM13},   \citet{AFHHM19}  and \citet{Z21}. Our characterization condition shares some similarity with a polymatroid as it requires only integer valued coefficients in linear inequalities. At the same time, it differs from a polymatroid in that the inequalities contain not only 0,1 coefficients but more general integer coefficients.

The two alternatives voting model has received attention in the literature in social choice theory -- from May's theorem \citep{M52} to
its extensions, including a recent extension by \citet{BHJTY21}. \citet{ST12} identify qualified majority rules as ex-ante welfare maximizing in the class of dominant strategy voting rules. The results in \citet{AK14} (which we discussed earlier) show that focusing attention to ordinal rules in this model is
without loss of generality in a certain sense -- see \citet{N04} also.

\newpage

\appendix

\section{Missing proofs}
\label{sec:app1}

We first prove Theorem \ref{theo:ext}, and then Theorem \ref{theo:main}.

\subsection{Proof of Theorem \ref{theo:ext}}

\begin{proof}
Reduced form probabilities $(Q(a),Q(b))$ is implementable if
\begin{align}
\frac{1}{n\pi}\sum_{k=0}^n k q(k)B(k) &= Q(a) \label{eq:ee1} \\
\frac{1}{n(1-\pi)}\sum_{k=0}^n (n-k)q(k)B(k) &= Q(b)  \label{eq:ee2} \\
0 \le q(k) &\le 1 \qquad~\forall~k \in \{0,1,\ldots,n\} \label{eq:ee3}
\end{align}
Let $\mathcal{P}$ be the projection of this polytope onto the $(Q(a),Q(b))$-space. Clearly, $\mathcal{P}$ is a polytope.
Consider the following linear program
\begin{align}
\max_{Q} \mu_a Q(a) + \mu_b Q(b) & \tag{\textbf{LP}-Q} \label{tag:q}\\
\textrm{subject to}~~~(Q(a),Q(b))~\in \mathcal{P} \nonumber
\end{align}
As we vary $\mu_a$ and $\mu_b$, the solutions to the linear program program (\ref{tag:q}) characterize
the {\sl boundary} points of $\mathcal{P}$.
Since each point in $\mathcal{P}$ is equivalent to
finding a voting rule $q$ that satisfies (\ref{eq:ee1}), (\ref{eq:ee2}), and (\ref{eq:ee3}),
we can rewrite the linear program (\ref{tag:q}) in the space of $q$ as:
\begin{align}
\max_{q} \Big[\frac{\mu_a}{n\pi} \sum_{k=0}^n k q(k)B(k) &+ \frac{\mu_b}{n(1-\pi)} \sum_{k=0}^n (n-k) q(k)B(k) \Big]
\tag{\textbf{LP}-q} \label{tag:q2} \\
\textrm{subject to}~~~0 \le q(k) &\le 1 \qquad~\forall~k \in \{0,1,\ldots,n\} \nonumber
\end{align}
Hence, the set of boundary points of $\mathcal{P}$ can be described
by the interim allocation probabilities of the voting rules
obtained as a solution to the linear program (\ref{tag:q2}) as we
vary $\mu_a$ and $\mu_b$. 

We now do the proof in two steps. \\

\noindent {\sc Step 1.} We first show that  every extreme point of $\mathcal{P}$ is implemented by either a qualified majority voting rule or a qualified anti-majority voting rule, i.e., every element of $\mathcal{P}$ can be written as a convex combination
of qualified (anti-)majority voting rules.
  
It is sufficient to show that for every $\mu_a$ and $\mu_b$, there is a solution to (\ref{tag:q}) that is implemented by either a qualified majority or a qualified anti-majority voting
rule.  To show this, we show that for every $\mu_a$ and $\mu_b$, some qualified (anti-)majority
voting rule is a solution to (\ref{tag:q2}).

By denoting $\hat{\mu}_a:= \mu_a/(n\pi)$ and $\hat{\mu}_b:= \mu_b/(n(1-\pi))$,
we see that the objective function of (\ref{tag:q2}) is
\begin{align*}
\sum_{k=0}^n \big[ n\hat{\mu}_b + k (\hat{\mu}_a - \hat{\mu}_b) \big]q(k)B(k)
\end{align*}
We show that  $n\hat{\mu}_b + k (\hat{\mu}_a - \hat{\mu}_b)  $ is either weakly increasing, in which case some qualified majority voting
rule is optimal) or weakly decreasing, in which case some qualified anti-majority voting rule
is optimal.  
 
If $n\hat{\mu}_b + k (\hat{\mu}_a - \hat{\mu}_b) > 0$ for all $k$, then a solution
to (\ref{tag:q2}) is to set $q(k)=1$
for all $k$. This is the qualified majority with quota $0$. If $n\hat{\mu}_b + k (\hat{\mu}_a - \hat{\mu}_b) < 0$
for all $k$, then a solution to (\ref{tag:q2}) is to set $q(k)=0$ for all $k$. This is the qualified anti-majority with quota $0$.
If $n\hat{\mu}_b + k (\hat{\mu}_a - \hat{\mu}_b) = 0$ for all $k$, then every voting rule $q$ is a solution.

If the sign of $n\hat{\mu}_b + k (\hat{\mu}_a - \hat{\mu}_b)$ changes with $k$, then we consider two
cases. If $\hat{\mu}_a > \hat{\mu}_b$, then there is a cut-off $k^*$ such that
$n\hat{\mu}_b + k (\hat{\mu}_a - \hat{\mu}_b) > 0$ for all $k \ge k^*$ and
$n\hat{\mu}_b + k (\hat{\mu}_a - \hat{\mu}_b) < 0$ for all $k < k^*$. Then, the
qualified majority with quota $k^*$ is a solution to (\ref{tag:q2}).
On the other hand if $\hat{\mu}_a < \hat{\mu}_b$, then there is a cutoff $k^*$
such that $n\hat{\mu}_b + k (\hat{\mu}_a - \hat{\mu}_b) > 0$ for all $k \le k^*$ and
$n\hat{\mu}_b + k (\hat{\mu}_a - \hat{\mu}_b) < 0$ for all $k > k^*$. Then, the qualified
anti-majority with quota $k^*$ is a solution of
(\ref{tag:q2}).\footnote{When $\hat{\mu}_a=\hat{\mu}_b$
the sign of $n\hat{\mu}_b + k (\hat{\mu}_a - \hat{\mu}_b)$ does not change with $k$.} Note that in both cases above, if $n\hat{\mu}_b + k (\hat{\mu}_a - \hat{\mu}_b) = 0$
for $k=k^*$,  the (anti-)qualified majority with quota $k^*$ is a solution to (\ref{tag:q2}). \\

\noindent {\sc Step 2.} We now show that every qualified (anti-)majority voting rule  implements a {\it distinct} extreme point of $\mathcal{P}$. Every extreme point in $\mathcal{P}$ is obtained by considering values of $\mu_a$ and $\mu_b$
which generate a {\sl unique} optimal solution to the linear program (\ref{tag:q}).
 It is sufficient to show that  every  qualified (anti-)majority voting rule is unique optimal solution
to (\ref{tag:q2}) for some $\mu_a$ and $\mu_b$. This is easily seen from our analysis above that for almost all $\mu_a$ and $\mu_b$, in case an optimal solution to (\ref{tag:q2}) exists, it is unique, corresponds to a qualified
majority or a qualified anti-majority voting rule.  \\

\noindent Combining Steps 1 and 2, we see that the set of extreme points of $\mathcal{P}$
 is the set of qualified majority voting rules and the set of qualified anti-majority voting rules.
 \end{proof}

\subsection{Proof of Theorem \ref{theo:main}}

\begin{proof}
We know that the necessary conditions for reduced form implementation are (\ref{eq:e1}) and (\ref{eq:e2}).
Let $\mathcal{P}^*$ denote the polytope described by (\ref{eq:e1}) and (\ref{eq:e2}).
We show that the extreme points of $\mathcal{P}^*$ correspond to
the qualified majority and the qualified anti-majority voting rules. From Theorem \ref{theo:ext},
we know that the extreme points of $\mathcal{P}$ also correspond to
the qualified majority and the qualified anti-majority voting rules. Hence, $\mathcal{P}=\mathcal{P}^*$.

To show that the extreme points of $\mathcal{P}^*$ correspond to
the qualified majority and the qualified anti-majority voting rules,
we follow two steps.

\noindent {\sc Every $q \in \mathcal{Q}^+ \cup \mathcal{Q}^-$ is an extreme point.}  
Consider any qualified majority voting rule with quota $j\in\{1,\ldots, n\}$. Using 
\begin{align*}
n\pi Q^j(a)  =  \sum_{k=j}^n k B(k)~~~\textrm{and}~~n(1-\pi)  Q^j(b)  =    \sum_{k=j}^{n} (n-k) B(k),
\end{align*} 
it is easy to verify that  $Q^j$ satisfies all  inequalities in  $(\ref{eq:e1})$ and $(\ref{eq:e2})$ and inequality
(\ref{eq:e1}) is binding for $j$ and $(j-1)$ at $Q^j$.  Since $Q^j\in\mathcal{P}^*$ and  $Q^j$ is the intersection of two linearly independent hyperplanes, it gives an extreme point of $\mathcal{P}^*$.  Since the qualified majority voting rule with quota $0$ corresponds to a
constant voting rule, that is also an extreme point.

An analogous argument shows that  the interim allocation probability of every qualified anti-majority voting rule with a quota $j \in \{0,\ldots,n\}$ is an extreme point. 

\noindent {\sc No extreme point outside $\mathcal{Q}^+ \cup \mathcal{Q}^-$.} Consider an extreme point of $\mathcal{P}^*$ that is not a  qualified  (anti-)majority rule. Then two non-adjacent constraints must be binding,  i.e., either (\ref{eq:e1}) binds for some
$j$ and $j+\ell$ with $\ell > 1$, or (\ref{eq:e2}) binds for some $j$ and $j+\ell$ with $\ell > 1$, or (\ref{eq:e1}) binds
for some $j$ and (\ref{eq:e2}) binds for some $\ell$. 

Assume first that (\ref{eq:e1}) binds for $j$ and $j+\ell$,
where $\ell > 1$. The equality corresponding to $(j+\ell)$ is
\begin{align*}
0 &= (j+\ell) (1-\pi)Q(b) - (n-j-\ell)\pi Q(a) + \sum_{k=j+\ell+1}^n (k-j-\ell) B(k) \\
&= \ell \Big( \pi Q(a) + (1-\pi) Q(b) \Big) + j (1-\pi)Q(b) - (n-j) \pi Q(a) \\
&+ \sum_{k=j+\ell+1}^n (k-j)B(k) - \sum_{k=j+\ell+1}^n\ell B(k)
\end{align*}
Since inequality  (\ref{eq:e1}) binds for $j$,  substitute the equality  into   (\ref{eq:e1}) for  $j+1$,
\begin{align*}
\pi Q(a) + (1-\pi)Q(b)  \ge \sum_{k=j+1}^n B(k)  
\end{align*}
We get
\begin{align*}
0 &\ge \sum_{k=j+1}^n \ell B(k)  - \sum_{k=j+\ell+1}^n \ell B(k) +  \sum_{k=j+\ell+1}^n (k-j)B(k) - \sum_{k=j+1}^n (k-j)B(k) \\
&= \sum_{k=j+1}^{j+\ell} \ell B(k) - \sum_{k=j+1}^{j+\ell} (k-j)B(k) = \sum_{k=j+1}^{j+\ell} (j+\ell - k) B(k) > 0
\end{align*}
which is a contradiction. Hence, (\ref{eq:e1}) cannot bind for $j$ and $(j+\ell)$ for $\ell > 1$.
An analogous proof shows that (\ref{eq:e2}) cannot bind for $j$ and $(j+\ell)$ for $\ell > 1$.

Now, assume   (\ref{eq:e1}) binds for $j$ and (\ref{eq:e2}) binds
for $\ell$. Hence, adding those two equalities, we get
\begin{align*}
0 &= (j-\ell) (1-\pi)Q(b) + (j-\ell) \pi Q(a) + \sum_{k=0}^{\ell-1} (\ell-k)B(k) + \sum_{k=j+1}^n (k-j)B(k)
\end{align*}
If $j \ge \ell$ and $(j,\ell)\neq (n,0)$, the RHS is positive, giving us a contradiction. If $j < \ell$ and $(j,\ell)\neq (0,n)$, using  $\pi Q(a) + (1-\pi) Q(b)\leq 1$, we get
\begin{align*}
0 &= (j-\ell) \Big((1-\pi)Q(b) + \pi Q(a)\Big) + \sum_{k=0}^{\ell-1} (\ell-k)B(k) + \sum_{k=j+1}^n (k-j)B(k) \\
&\ge j-\ell+ \sum_{k=0}^{\ell-1} (\ell-k)B(k) + \sum_{k=j+1}^n (k-j)B(k) \\
&= j\Big(1-\sum_{k=j+1}^n B(k)\Big)- \ell \Big(1-\sum_{k=0}^{\ell-1}B(k)\Big)+\Big(\sum_{k=\ell }^n k B(k)-n\pi\Big)+   \Big(n\pi-\sum_{k=0}^{j } kB(k)\Big)\\
&=\sum_{k=0}^j  (j-k)  B(k)  +\sum_{k=\ell }^n (k- \ell) B(k)  >0
\end{align*}
which also gives us a contradiction.

 If $(j,\ell)=(n,0)$ or $ (0,n)$, the two equalities determine $(Q(a),Q(b))=(0,0)$ or $ (1,1)$, which correspond to the two constant voting rules, which are in $\mathcal{Q}^+ \cap \mathcal{Q}^-$.
\end{proof}

\subsection{Proof of Theorem \ref{theo:monimp}}

\begin{proof}
\noindent $(1) \Rightarrow (2)$. Since $Q$ is reduced form monotone implementable,
it is reduced form implementable by a monotone voting rule $q$. Hence, we can
write
\begin{align*}
n\pi (1-\pi) \Big[ Q(a) - Q(b) \Big] &= \sum_{k=0}^n \big[ k (1-\pi) - (n-k)\pi\big]q(k)B(k) =  \sum_{k=0}^n (k-n\pi) q(k)B(k) \\
&\ge q(\lfloor n\pi \rfloor) \sum_{k=0}^n (k-n\pi)B(k) = 0
\end{align*}
where we use monotonicity of $q$ for the inequality.
This shows $Q(a) \ge Q(b)$. \\

\noindent $(2) \Rightarrow (3)$. If $Q$ is reduced form implementable, by Theorem \ref{theo:ext},
it can be expressed as convex combination of interim allocation probabilities of
qualified majority and qualified anti-majority voting rules.

Consider any qualified anti-majority with quota $j\in \{0,\ldots,n\}$ (qualified anti-majority
with quota $0$ corresponds to a constant voting rule).
Define for each $j\in \{0,\ldots,n\}$
\begin{align*}
\delta(j):=\overline{Q}^j(a) - \overline{Q}^j(b) &= \frac{1}{n\pi}\sum_{k=0}^{j-1} k B(k) - \frac{1}{n(1-\pi)}\sum_{k=0}^{j-1} (n-k) B(k) = \frac{1}{n\pi(1-\pi)}\sum_{k=0}^{j-1} (k-n\pi)B(k)
\end{align*}
Note that $\delta(0)=0$ and $\delta(n) = -n(1-\pi)B(n) < 0$.

For all $j \in \{0,\ldots,n-1\}$, we get 
\begin{align*}
\delta (j+1) - \delta(j) &= \frac{1}{n\pi(1-\pi)} (j-n\pi)B(j)
\end{align*}
which is non-negative if $j \ge n\pi$ and negative if $j < n\pi$. Hence, value of $\delta(j)$ decreases with $j$
for all $j < n\pi$ and increases after that till $j=n$. Since $\delta(0)=0$ and $\delta(n) < 0$, we conclude
that $\delta(j) = \overline{Q}^j(a) - \overline{Q}^j(b) < 0$ for all $j \in \{1,\ldots,n\}$ and $\delta(0)=0$. 

On the other hand, for any qualified majority with quota $j$, we have $Q^j(a) \ge Q^j(b)$.
The qualified anti-majority with quota zero corresponds to
a constant voting rule which generates interim allocation probabilities $Q(a)=Q(b)=0$.
Hence, if $Q(a) \ge Q(b)$, then $Q$ is reduced form implementable by convex combination of
qualified majority voting rules and a constant voting rule that selects $b$ at all type profiles. \\

\noindent $(3) \Rightarrow (4)$. Every qualified majority and qualified anti-majority with quota zero
generates interim allocation probabilities $Q$ that satisfy $Q(a) \ge Q(b)$.
Hence, their convex combination also satisfies $Q(a) \ge Q(b)$. By Theorem \ref{theo:main},
if $Q$ is reduced form implementable then it satisfies (\ref{eq:mm1}). \\

\noindent $(4) \Rightarrow (1)$. The proof of Theorem \ref{theo:main} shows that the
set of extreme points of (\ref{eq:mm1}) is the set of qualified majority voting rules.
The line $Q(a)=Q(b)$ connects two constant voting rules and all the qualified majority voting rules
satisfy $Q(a) \ge Q(b)$. As a result, any $Q$ satisfying (\ref{eq:mm1}) and (\ref{eq:mm2})
must be reduced-form equivalent to a convex combination of qualified majority voting rules and
the two constant voting rules. Hence, it is reduced form monotone implementable.
\end{proof}

\subsection{Proof of Proposition \ref{prop:rawls}}

\begin{proof}
By Theorem \ref{theo:monimp}, the ex-ante Rawlsian rule solves the following optimization problem
\begin{align}
\max_Q \min \big(\pi Q(a), (1-\pi) (1-Q(b)\big) & \notag\\
\textrm{subject to}~~~~Q(a) &\ge Q(b)  \label{eq:M1}\\
j(1-\pi)Q(b) - (n-j)\pi Q(a) + \sum_{k=j}^n (k-j) B(k) &\ge 0 ~\qquad~\forall~j \in \{0,\ldots,n\}\label{eq:M2}
\end{align}
 
Consider the relaxed problem where we drop the inequalities in (\ref{eq:M1}). Further,
change the variables as follows: $x:=\pi Q(a)$ and $y:=(1-\pi)(1-Q(b))$.
So, the relaxed problem (with inequalities (\ref{eq:M1}) in terms of $x,y$) is
the following
\begin{align}
\max_{x,y} \min \big(x,y\big) & \notag\\
\textrm{subject to}~~~~j y + (n-j)x &\le j(1-\pi) + \sum_{k=j}^n (k-j) B(k)~\qquad~\forall~j \in \{0,\ldots,n\}\label{eq:M21}
\end{align}
Notice that for any feasible solution $(x,y)$ to the above problem, the
solution $\hat{x}=\hat{y}=\min (x,y)$ is also a feasible solution with 
the same objective function value. Hence, it is without loss of generality
to assume $x=y$. Hence, substituting $x=y$ on the LHS of (\ref{eq:M21}), 
we get $nx$, and the problem simplifies to 
\begin{align}
\max_{x} x & \notag\\
\textrm{subject to}~~nx &\le j(1-\pi) + \sum_{k=j}^n (k-j) B(k)~\qquad~\forall~j \in \{0,\ldots,n\}\label{eq:M22}
\end{align}
For every $j \in \{0,\ldots,n\}$, let $H(j):= j(1-\pi) + \sum_{k=j}^n (k-j) B(k)$.
Hence, the optimal solution is given by
\begin{align*}
x = y &= \frac{1}{n}\min_{j \in \{0,\ldots,n\}} H(j)
\end{align*}
For $j \in \{1,\ldots,n\}$, we see 
\begin{align*}
H(j) - H(j-1) = 1 - \pi - \sum_{k=j}^n B(k)
\end{align*}
Let $j^*:=\max \{j\in\{0,\cdots,n\}: \sum \limits_{k=j}^n  B(k) \geq 1-\pi\}$. Then, $H$ is
decreasing till $j^*$ and increasing after that. 
So, $x=y=(1/n)H(j^*)$ is an optimal
solution to the relaxed problem. This optimal solution corresponds to 
\begin{align*}
Q(a) &= \frac{1}{n\pi} \Big[  j^*(1-\pi) + \sum_{k=j^*}^n (k-j^*) B(k) \Big] \\
Q(b) &=  \frac{1}{n(1-\pi)} \Big[  (n-j^*)(1-\pi) - \sum_{k=j^*}^n (k-j^*) B(k) \Big]
\end{align*}
This corresponds to satisfying inequality (\ref{eq:M22}) for $j^*$.

  Now, define 
\begin{align*}
\alpha :=\frac{1}{B(j^*)}\left(1-\pi-\sum_{k=j^*+1}^n B(k)\right)
\end{align*}
By definition of $j^*$, $\alpha \in [0,1]$. Using the expressions
for $Q^{j^*}(a)$ and $Q^{j^*+1}(a)$, it can be easily verified that
\begin{align*}
Q(a)=\alpha  Q^{j^*}(a)+(1-\alpha  )Q^{j^*+1}(a) \\
Q(b)=\alpha  Q^{j^*}(b)+(1-\alpha  )Q^{j^*+1}(b)
\end{align*}
This shows that the optimal $Q$ is a convex combination of two qualified
majority voting rules with quotas $j^*$ and $j^*+1$.

Since each qualified majority is monotone,
$Q$ is also monotone. Hence, the optimum of the relaxed problem
is a monotone voting rule. \end{proof}

\subsection{Proofs of Propositions \ref{prop:large1} and \ref{prop:large2}}

\noindent {\sc Proof of Proposition \ref{prop:large1}.}

\begin{proof}
We keep $\pi$ fixed and make $n$ large. By Theorem \ref{theo:ext}, it is enough to show that for each
qualified majority $Q^j$ with quota $j$ (and qualified anti-majority) the difference in interim allocation
probabilities $Q^j(a;n) - Q^j(b;n)$ approaches zero as $n$ tends to infinity. Note that when
$j=0$, $Q^j(a;n)=Q^j(b;n)$. Hence, we only consider the case $j > 1$. By (\ref{eq:diffq}),
\begin{align}
Q^j(a;n) - Q^j(b;n)  = \frac{j}{n\pi}  C(n ,j  ) \pi^{j} (1-\pi)^{(n-j)}\leq \frac{1}{\pi}  C(n ,j  ) \pi^{j} (1-\pi)^{(n-j)} \label{eq:diffq1}
\end{align}
For $n$ sufficiently large,  the probability mass  of the Binomial distribution approaches the probability density of the normal distribution with mean $n\pi$ and variance $n\pi (1-\pi)$. Denoting the density function of this normal
distribution as $f$, we have for each $j=0,...,n$,
\begin{align*}
   C(n,j) \pi^{j} (1-\pi)^{(n-j)}\approx f(j;  n\pi,  n\pi (1-\pi) )
   \end{align*}
   The maximum of the probability mass function is obtained at $j= \lfloor (n+1)\pi  \rfloor$,
\begin{align*}
\max_{j\in\{0,...,n\}} C(n,j) \pi^{j} (1-\pi)^{(n-j)}\approx  f( \lfloor (n+1)\pi  \rfloor;  n\pi,  n\pi (1-\pi) )
\end{align*}
Notice that for all $n$, $  \lfloor (n+1)\pi  \rfloor-n\pi\leq 1$ and we have
\begin{align*}
  \lim_{n\to \infty} f( \lfloor (n+1)\pi  \rfloor;  n\pi,  n\pi (1-\pi) )   =  \lim_{n\to \infty} \frac{\exp \left({-\frac{1}{2}\left(\frac{ \lfloor (n+1)\pi  \rfloor-n\pi}{ \sqrt{n\pi (1-\pi)}}\right)^2} \right)}{\sqrt{2 \Pi }\sqrt{n\pi (1-\pi)}} =0,
\end{align*}
where $\Pi$ denotes the usual mathematical constant.\footnote{To avoid notational confusion, we use $\Pi$
instead of $\pi$ to denote the ratio of circumference of a circle and its diameter.}
Therefore, (\ref{eq:diffq1}) implies for every $j$, we have
\begin{align*}
\lim_{n \rightarrow \infty} \big[Q^j(a;n) - Q^j(b;n)\big] \leq   \lim_{n\to \infty}\max_{j\in\{0,...,n\}} C(n,j) \pi^{j} (1-\pi)^{(n-j)}=0
\end{align*}
 Since  $Q^j(a;n)-Q^j(b;n)\geq 0$, we conclude
\begin{align*}
\lim_{n \rightarrow \infty} \big[Q^j(a;n) - Q^j(b;n)\big] =  0
\end{align*}
Using (\ref{eq:diffaq}), we get that for every qualified anti-majority rules with quota $j > 1$
\begin{align*}
\lim_{n \rightarrow \infty} \big[\overline{Q}^{j}(b;n) - \overline{Q}^{j}(a;n)\big] =0
\end{align*}
\end{proof}

\noindent {\sc Proof of Proposition \ref{prop:large2}.}

\begin{proof}
Fix the mean $\mu$ and take a sequence of economies where $\pi_n$ such that $\pi_n = \mu/n$. Here, $\pi_n$ denotes the value of $\pi$ in an economy with $n$ agents.
By the Poisson limit theorem,
\begin{align*}
\lim_{n \rightarrow \infty} C(n,j) \pi_n^{j} (1-\pi_n)^{(n-j)} &= \frac{1}{j!}\mu^j e^{-\mu}
\end{align*}
Hence, using (\ref{eq:diffq}), for any qualified majority with quota $j > 1$, we have
\begin{align*}
\lim_{n \rightarrow \infty} \big[Q^j(a;n) - Q^j(b;n)\big] &= \frac{1}{(j-1)!}\mu^{j-1} e^{-\mu}
\end{align*}

Let $k_{\mu}$ be the value of $k$ that maximizes
\begin{align*}
\max_{k \in \mathbb{Z}_+} \frac{\mu^k}{k!}
\end{align*}
Note that a maximum exists since as $k \rightarrow \infty$, the expression $\mu^k/(k!)$ tends to zero.
So $k_{\mu}$ is a finite integer. Denote this maximum value multiplied by $e^{-\mu}$ as $M(\mu):= \frac{1}{(k_{\mu})!}\mu^{k_{\mu}}e^{-\mu}$.

Hence, we get
\begin{align}
\lim_{n \rightarrow \infty} \big[Q^{k_{\mu}+1}(a;n) - Q^{k_{\mu}+1}(b;n)\big] &= M(\mu) \label{eq:m1}
\end{align}

Now, for anti-majority rule with quota $j > 1$,
by (\ref{eq:diffaq}), we get
\begin{align*}
\lim_{n \rightarrow \infty} \big[\overline{Q}^j(b;n) - \overline{Q}^j(a;n)\big] &= \frac{1}{(j-1)!}\mu^{j-1} e^{-\mu}
\end{align*}
Hence, we get
\begin{align}
\lim_{n \rightarrow \infty} \big[\overline{Q}^{k_{\mu}+1}(b;n) - \overline{Q}^{k_{\mu}+1}(a;n)\big] &= M(\mu) \label{eq:m2}
\end{align}
Equations (\ref{eq:m1}) and (\ref{eq:m2}) proves the proposition.
\end{proof}

\subsection{Proof of Proposition \ref{prop:icf}}

\begin{proof}
By Theorem \ref{theo:uext}, every unanimous voting rule is reduced
form equivalent to a convex combination of u-qualified majority and u-qualified anti-majority rules.
Since a convex combination preserves OBIC, every unanimous voting rule
is OBIC if and only if every u-qualified majority and u-qualified anti-majority rule is OBIC.
We know that every u-qualified majority is OBIC (since they are strategy-proof).
Hence, every unanimous voting rule
is OBIC if and only if every u-qualified anti-majority rule is OBIC.

Let $\bar{q}^j $ be a u-qualified anti-majority rule with quota $j\in\{1,\ldots,n\}$.
Then,
\begin{align*}
\bar{Q}^j(a|a) = \bar{Q}^j(a) &= \frac{1}{n\pi} \Big[\sum_{k=1}^{j-1} k B(k) + n B(n) \Big] \\
\bar{Q}^j(b|b) = \bar{Q}^j(b) &= \frac{1}{n(1-\pi)} \sum_{k=1}^{j-1} (n-k)B(k)
\end{align*}

The value of $\bar{Q}^j(b|a)$ is computed as follows:
 \begin{align*}
\bar{Q}^j(b|a) &=\frac{1}{ \pi}\sum_{k=0}^{n-1} q^j (k )\lambda(k+1)C(n-1,k)=\frac{1}{n \pi}\sum_{k=0}^{n-1}q^j (k )\lambda(k+1)(k+1)C(n,k+1)\\
 &=\frac{1}{n \pi}\sum_{k=1}^{n}  q^j (k-1)kB(k) =\frac{1}{n \pi}\sum_{k=2}^{j} kB(k)
 \end{align*}
  Similarly we have
   \begin{align*}
 \bar{Q}^j(a|b) &=\frac{1}{1-\pi}\sum_{k=0}^{n-1} q^j (k+1)\lambda(k)C(n-1,k) =\frac{1}{n(1-\pi)}\sum_{k=0}^{n-1} q^j (k+1)\lambda(k) (n-k)C(n, k)\\
 &=\frac{1}{n(1-\pi)}\Big(\sum_{k=0}^{n-2}q^j (k+1)  \lambda(k) (n-k)C(n, k)+n\lambda(n-1)\Big) \\
 &=\frac{1}{n(1-\pi)}\Big(\sum_{k=0}^{j-2}    (n-k)B(k)+n\lambda(n-1)\Big)
\end{align*}

 Hence,
 \begin{align*}
n\pi[\bar{Q}^j(a|a) - \bar{Q}^j(b|a)] &=  \sum_{k=0}^{j -1} k B(k)+n\lambda(n)   -  \sum_{k=2}^{j} kB(k) = B(1) - j B(j) + n \lambda(n)
\end{align*}
So, $\bar{Q}^j(a|a) - \bar{Q}^j(b|a) \ge 0$ if and only if $n (\lambda(1)+\lambda(n)) \ge j B(j)$.
This inequality trivially holds for $j=1$ and $j=n$. Hence, the inequality needs to hold for
all $j \in \{2,\ldots,n-1\}$.
Similarly,
 \begin{align*}
n(1-\pi)[\bar{Q}^j(a|b) - \bar{Q}^j(b|b)] &= \sum_{k=0}^{j-2} (n-k)B(k)+n\lambda(n-1) - \sum_{k=0}^{j-1} (n-k) B(k)+ n\lambda(0) \\
&= n (\lambda(n-1) + \lambda(0)) - (n-j+1)B(j-1)
\end{align*}
Hence, $\bar{Q}^j(a|b) - \bar{Q}^j(b|b) \ge 0$ if and only if $n (\lambda(n-1) + \lambda(0))  \ge (n-j+1)B(j-1)$.
Hence, $n (\lambda(n-1) + \lambda(0))  \ge (n-j)B(j)$ should hold for $j \in \{0,1,\ldots,n-1\}$.
This inequality holds for $j=n-1$ and $j=0$ trivially.
Note that  $j B(j)= n\lambda(j) C(n-1,j-1)$  and
  $(n-j) B(j)=n \lambda(j) C(n-1,j)  $. Then we obtain condition (\ref{eq:icf1}).

When the prior is independent, (\ref{eq:icf1}) is equivalent to (\ref{eq:icf2}). To see this, pick $j \in \{1,\ldots,n-1\}$,
  \begin{align*}
n( \lambda(1)+\lambda(n) ) & \ge  j B(j)    \\
\Leftrightarrow n \pi ((1-\pi)^{n-1}+\pi^{n-1}) &\ge j B(j)
\end{align*}

Next,
\begin{align*}
n( \lambda(0) + \lambda(n-1))& \ge (n-j)B(j)  \\
\Leftrightarrow n (1-\pi) ((1-\pi)^{n-1}+\pi^{n-1})& \ge n\pi^j (1-\pi)^{n-j} C(n-1,j) \\
\Leftrightarrow n \pi ((1-\pi)^{n-1}+\pi^{n-1}) &\ge (j+1) B(j+1)
\end{align*}
Hence, for independent priors, condition (\ref{eq:icf1}) is equivalent to for all $j \in \{1,\ldots,n-1\}$,
\begin{align*}
n \pi ((1-\pi)^{n-1}+\pi^{n-1}) &\ge j B(j)
\end{align*}
This is equivalent to (\ref{eq:icf2}).
\end{proof}


\newpage

\section{Supplementary Appendix}
\label{sec:supapp}

Proofs of Theorem \ref{theo:unan}
 and Theorem \ref{theo:uext} are similar to Theorem \ref{theo:main} and Theorem \ref{theo:ext} respectively.
 They are provided here for completeness.

\subsection{Proof of Theorem \ref{theo:uext}}

\begin{proof}
Reduced form probabilities $(Q(a),Q(b))$ is u-implementable if
\begin{align}
\frac{1}{n\pi} \Big[\sum_{k=1}^{n-1}k q(k) B(k) + n B (n)\Big] &= Q(a)\label{eq:ue1} \\
\frac{1}{n(1-\pi)} \sum_{k=1}^{n-1}(n-k) q(k)B(k) &= Q(b) \label{eq:ue2}\\
0 \le q(k) &\le 1 \qquad~\forall~k \in \{1,\ldots,n-1\} \label{eq:ue3}
\end{align}

Let $\mathcal{P}_u$ be the projection of this polytope to the $(Q(a),Q(b))$ space.
Consider the following linear program
\begin{align}
\max_{Q} \mu_a Q(a) + \mu_b Q(b) & \tag{u\textbf{LP}-Q} \label{tag:uq}\\
\textrm{subject to}~~~(Q(a),Q(b))~\in \mathcal{P}_u \nonumber
\end{align}

Since each point in $\mathcal{P}_u$ is equivalent to
finding a voting rule $q$ that satisfies (\ref{eq:ue1}), (\ref{eq:ue2}), and (\ref{eq:ue3})
we can rewrite the linear program (\ref{tag:uq}) in the space of $q$ as:
\begin{align}
\max_{q} \frac{\mu_a}{n\pi}   \Big[\sum_{k=1}^{n-1}k q(k) B(k) &+ n B (n)\Big] + \frac{\mu_b}{n(1-\pi)}    \sum_{k=1}^{n-1}(n-k) q(k)B(k)
\tag{u\textbf{LP}-q} \label{tag:uq2} \\
\textrm{subject to}~~~0 \le q(k) &\le 1 \qquad~\forall~k \in \{1,\ldots, n-1\} \nonumber
\end{align}

We do the proof in two steps. \\

\noindent {\sc Step 1.} We first show that  every extreme point of $\mathcal{P}_u$ is implemented by either a  u-qualified majority or a u-qualified anti-majority voting rule, i.e., every element of $\mathcal{P}_u$ can be written as a convex combination
of u-qualified (anti-)majority voting rules.
  
It is sufficient to show that for every $\mu_a$ and $\mu_b$, there is a solution to (\ref{tag:uq}) that is implemented by some u-qualified (anti-)majority voting rule.  To show this, we will show that for every $\mu_a$ and $\mu_b$, some u-qualified (anti-)majority voting rule is a solution to (\ref{tag:uq2}).
  
By denoting $\hat{\mu}_a:= \mu_a/(n\pi)$ and $\hat{\mu}_b:= \mu_b/(n(1-\pi))$,
we see that the objective function of (\ref{tag:uq2}) is
\begin{align*}
\sum_{k=1}^{n-1} \big[ n\hat{\mu}_b + k (\hat{\mu}_a - \hat{\mu}_b) \big]q(k)B(k)+\hat{\mu}_a n B(n)
\end{align*}

If $n\hat{\mu}_b + k (\hat{\mu}_a - \hat{\mu}_b) > 0$ for all $k$, then a solution
to (\ref{tag:uq2}) is to set $q(k)=1$
for all $k$. This is the qualified majority with quota $0$. If $n\hat{\mu}_b + k (\hat{\mu}_a - \hat{\mu}_b) < 0$
for all $k$, then a solution to (\ref{tag:uq2}) is to set $q(k)=0$ for all $k$. This is the u-qualified anti-majority with quota $1$.
If $n\hat{\mu}_b + k (\hat{\mu}_a - \hat{\mu}_b) = 0$ for all $k$, then  every unanimous rule is a solution to (\ref{tag:uq}).

If the sign of $n\hat{\mu}_b + k (\hat{\mu}_a - \hat{\mu}_b)$ changes with $k$, then we consider two
cases. If $\hat{\mu}_a > \hat{\mu}_b$, then there is a cut-off $k^*$ such that
$n\hat{\mu}_b + k (\hat{\mu}_a - \hat{\mu}_b) > 0$ for all $k \ge k^*$ and
$n\hat{\mu}_b + k (\hat{\mu}_a - \hat{\mu}_b) < 0$ for all $k < k^*$. Then, the
 qualified majority with quota $k^*$ is a solution to (\ref{tag:uq2}).
On the other hand if $\hat{\mu}_a < \hat{\mu}_b$, then there is a cutoff $k^*$
such that $n\hat{\mu}_b + k (\hat{\mu}_a - \hat{\mu}_b) > 0$ for all $k \le k^*$ and
$n\hat{\mu}_b + k (\hat{\mu}_a - \hat{\mu}_b) < 0$ for all $k > k^*$. Then, the u-qualified
anti-majority with quota $k^*+1$ is a solution. When $\hat{\mu}_a=\hat{\mu}_b$
the sign of $n\hat{\mu}_b + k (\hat{\mu}_a - \hat{\mu}_b)$ does not change with $k$.
Note that in both cases above, if $n\hat{\mu}_b + k (\hat{\mu}_a - \hat{\mu}_b) = 0$
for $k=k^*$,  the u-qualified (anti-)majority with quota $k^*$ is a solution to (\ref{tag:uq2}). \\

\noindent {\sc Step 2.} We now show that every u-qualified (anti-)majority voting rule  implements a {\it distinct} extreme point of $\mathcal{P}_u $. Every extreme point in $\mathcal{P}_u$ is obtained by considering values of $\mu_a$ and $\mu_b$
which generate a {\sl unique} optimal solution to the linear program (\ref{tag:uq}).
 It is sufficient to show that  every  u-qualified (anti-)majority voting rule is unique optimal solution
to (\ref{tag:uq2}) for some $\mu_a$ and $\mu_b$. This can be seen from the analysis above that for almost all $\mu_a$ and $\mu_b$, in case an optimal solution to (\ref{tag:uq2}) exists, it is unique,  and corresponds to a u-qualified majority or a u-qualified
anti-majority voting rule.

Combining Steps 1 and 2,  we have that the set of extreme points of $\mathcal{P}_u$
 are the set of u-qualified majority voting rules and the set of u-qualified anti-majority voting rules. \end{proof}

\subsection{Proof of Theorem \ref{theo:unan}}

\begin{proof}
\noindent {\sc Necessity.} The necessity of (\ref{eq:u1}) follows from (\ref{eq:e1}) in Theorem \ref{theo:main}.
So, we only show necessity of (\ref{eq:u2}). Suppose $Q$ is reduced form u-implementable by
a unanimous voting rule $q$:
\begin{align*}
\frac{1}{n\pi} \Big[\sum_{k=1}^{n-1}k q(k) B(k) + n B (n)\Big] &= Q(a) \\
\frac{1}{n(1-\pi)} \sum_{k=1}^{n-1}(n-k)q(k)B(k) &=Q(b) \\
\end{align*}
Now, pick $j \in \{0,\ldots,n\}$ and observe that
\begin{align*}
n(n-j)\pi Q(a) - n j (1-\pi)Q(b) &= (n-j) \big[\sum_{k=1}^{n-1}k q(k) B(k) + n B(n) \big] - j \sum_{k=1}^{n-1}(n-k) q(k)B(k) \\
&= n(n-j) B(n) - n j \sum_{k=1}^{n-1}q(k)B(k) + n \sum_{k=1}^{n-1}k q(k) B(k)
\end{align*}
Hence, we have
\begin{align*}
(n-j)\pi Q(a) - j (1-\pi)Q(b) &= (n-j) B(n) - \sum_{k=1}^{n-1}(j-k) q(k) B(k)
\end{align*}
Hence,
\begin{align*}
(n-j)\pi Q(a) - j (1-\pi)Q(b) +  \sum_{k=1}^{j}(j-k) B(k) &\ge (n-j) B(n) \\
\Rightarrow (n-j)\pi Q(a) - j (1-\pi)Q(b) +  \sum_{k=0}^{j}(j-k) B(k) &\ge j B(0) + (n-j) B(n) = j \lambda(0) + (n-j) \lambda(n)
\end{align*}

\noindent {\sc Sufficiency.} Let $\mathcal{P}^*_u$ denote the polytope described by (\ref{eq:u1}) and (\ref{eq:u2}).
We show that the extreme points of $\mathcal{P}^*_u$ correspond to
the u-qualified majority and the u-qualified anti-majority voting rules. From Theorem 4,
we know that the extreme points of $\mathcal{P}_u$ also correspond to
the u-qualified majority and the u-qualified anti-majority voting rules. Hence, $\mathcal{P}_u=\mathcal{P}^*_u$.

To show that the extreme points of $\mathcal{P}^*_u$ correspond to
the u-qualified majority and the u-qualified anti-majority voting rules,
we follow two steps. \\

\noindent {\sc Every $q \in \mathcal{Q}_u^+ \cup \mathcal{Q}_u^-$ is an extreme point.}  Consider any u-qualified majority voting rule with quota $j\in\{1,\ldots, n\} $. Using
\begin{align*}
n\pi Q^+_u(a)  =  \sum_{k=j}^n k B(k)~~~\textrm{and}~~~n(1-\pi)Q^+_u(b)  =    \sum_{k=j}^{n} (n-k) B(k),
\end{align*} 
it is easy to verify that  $Q^+_u$ satisfies all  inequalities in  $(\ref{eq:u1})$ and $(\ref{eq:u2})$ and inequality
(\ref{eq:u1}) is binding for $j$ and $(j-1)$ at $Q^+_u$.  Since $Q^+_u\in\mathcal{P}^*_u$ and  $Q^+_u$ is the intersection of two linearly independent hyperplanes, it gives an extreme point of $\mathcal{P}^*_u$.  Hence the interim allocation probability of every u-qualified majority voting rule  is an extreme point.

An analogous argument shows that  the interim allocation probability of every u-qualified anti-majority voting rule   is an extreme point.
  
\noindent {\sc No extreme point outside $\mathcal{Q}_u^+ \cup \mathcal{Q}_u^-$.} Analogous to the proof of Theorem \ref{theo:main}, we can show that inequality (\ref{eq:u1}) cannot bind for $j$ and $(j+\ell)$ for $\ell > 1$. Now assume for contradiction that inequality (\ref{eq:u2}) binds for $j$ and $j+\ell$,
where $\ell > 1$. The equality corresponding to $(j+\ell)$ is
\begin{align*}
0 =& (n-j-\ell)\pi Q(a)-(j +\ell)(1-\pi)Q(b) + \sum_{k=0}^{j+\ell}(j+\ell-k) B(k)-(j+\ell)\lambda(0)-(n-j-\ell)\lambda(n)
 \end{align*}
Since inequality  (\ref{eq:u2}) binds for $j$, substitute this equality into inequality  (\ref{eq:u2}) for $(j+1)$,  it  gives us
\begin{align*}
\pi Q(a) + (1-\pi)Q(b) &\le \sum_{k=0}^j B(k)-\lambda(0)+\lambda(n) \label{eq:u23}
\end{align*}
Then we get
\begin{align*}
0 &\ge   -\sum_{k=0}^j (j-k) B(k) +j\lambda(0)+(n-j)\lambda(n) - \sum_{k=0}^j\ell B(k) + \ell \lambda(0)-\ell \lambda(n)\\
&\,\,\,\,\,\,\,+ \sum_{k=0}^{j+\ell}(j+\ell-k)B(k) - (j+\ell)\lambda(0)-(n-j-\ell)\lambda(n)\\
 &= \sum_{k=j+1}^{j+\ell} (j+\ell - k) B(k) \\
&> 0
\end{align*}
which is a contradiction. Hence, inequality (\ref{eq:u2}) cannot bind for $j$ and $(j+\ell)$ for $\ell > 1$.

Next assume for contradiction inequality (\ref{eq:u1}) binds for $j$ and inequality (\ref{eq:u2}) binds
for $\ell$. Hence, adding those two equalities, we get
\begin{align*}
0&= (j-\ell) ((1-\pi)Q(b)+\pi Q(a))    + \sum_{k=0}^{\ell } (\ell-k)B(k) + \sum_{k=j }^n (k-j)B(k)-\ell \lambda(0)-(n-\ell)\lambda(n)
\end{align*}

If $1\leq  \ell\leq j\leq n-1$ and $(j,\ell)\neq (n-1,1)$,   using $(1-\pi)Q(b) + \pi Q(a)\geq \lambda(n)$,   we get
\begin{align*}
0 & \ge (j-\ell)\lambda(n)+  \sum_{k=0}^{\ell } (\ell-k)B(k) + \sum_{k=j }^n (k-j)B(k)-\ell \lambda(0)-(n-\ell)\lambda(n) \\
&= \ell \lambda(0)+    \sum_{k=1}^{\ell } (\ell-k)B(k)+  \sum_{k=j }^{n-1} (k-j)B(k)+(n-j)\lambda(n)-\ell \lambda(0)-(n-j)\lambda(n)     \\
&=  \sum_{k=1}^{\ell } (\ell-k)B(k)+  \sum_{k=j }^{n-1} (k-j)B(k) \\
&>0
\end{align*}

If $1\leq j < \ell\leq n-1$ and $(j,\ell)\neq (1,n-1)$, using $(1-\pi)Q(b) + \pi Q(a)\leq 1- \lambda(0)$, we get
\begin{align*}
0 &  \ge    (j-\ell)-   (j-\ell)\lambda(0) + \sum_{k=0}^{\ell-1} (\ell-k)B(k) + \sum_{k=j+1}^n (k-j)B(k)-\ell \lambda(0)-(n-\ell)\lambda(n)  \\
 &=\sum_{k=0}^j  (j-k)  B(k)  +\sum_{k=\ell }^n (k- \ell) B(k)- j\lambda(0)-(n-\ell)\lambda(n)    \\
  &=j\lambda(0)+\sum_{k=1}^j  (j-k)  B(k)  +\sum_{k=\ell }^{n-1} (k- \ell) B(k)+ (n-\ell)\lambda(n) - j\lambda(0)-(n-\ell)\lambda(n)\\
  &= \sum_{k=1}^j  (j-k)  B(k)  +\sum_{k=\ell }^{n-1} (k- \ell) B(k)   \\
 &>0
\end{align*}
which also gives us a contradiction. On the other hand, for $(j,\ell)= (n-1,1)$, we have
 \begin{align*}
(n-1) (1-\pi)Q(b) -  \pi Q(a)+B(n) &= 0\\
(n-1)\pi Q(a) -  (1-\pi)Q(b)  +  B(0) &=  \lambda(0)+(n-1)\lambda(n)
\end{align*}
which gives  $Q(b)=0$ and $\pi Q(a)=\lambda(n)$, corresponding to a u-qualified majority with quota $n$. Analogously, for $(j,\ell)= (1,n-1)$, (\ref{eq:u1}) and  (\ref{eq:u2})
   give $Q(a)=1$ and $ (1-\pi)Q(b) +\pi=1-\lambda(0)$, which corresponds to a u-qualified anti-majority with quota $n$.  For $j=0,n$ and $\ell=0,n$, the inequalities are implied by $(n-1,1)$ and $(1,n-1)$ and hence redundant.
\end{proof}

\end{document}